\pdfoutput=1
\documentclass[11pt]{article}

\usepackage[utf8]{inputenc}
\usepackage[T1]{fontenc}
\usepackage{lmodern}
\usepackage{microtype}

\usepackage[letterpaper, margin=1in]{geometry}
\usepackage{indentfirst}

\usepackage{amsmath, amssymb, amsthm}
\usepackage{mathrsfs}
\usepackage{mathtools}

\usepackage{algorithm}
\usepackage{algpseudocode}
\algrenewcommand\algorithmicrequire{\textbf{Input:}}
\algrenewcommand\algorithmicensure{\textbf{Output:}}
\usepackage{graphicx}
\usepackage{subcaption}
\usepackage{booktabs}
\usepackage{enumitem}
\usepackage{xcolor}

\usepackage{authblk}

\definecolor{linkblue}{RGB}{0,51,153}
\usepackage[
    colorlinks=true,
    linkcolor=linkblue,
    citecolor=linkblue,
    urlcolor=linkblue,
    bookmarksnumbered=true,
    bookmarksopen=true,
]{hyperref}

\theoremstyle{plain}
\newtheorem{theorem}{Theorem}[section]
\newtheorem{lemma}[theorem]{Lemma}

\theoremstyle{definition}

\theoremstyle{remark}

\newcommand{\R}{\mathbb{R}}

\title{On the Cardinality of Optimal Representations in the Binary-Source Information Bottleneck}

\author[1]{Dier Tang\thanks{Email: \texttt{tangde\_math@connect.hku.hk}}}
\author[2]{Jun Chen\thanks{Email: \texttt{chenjun@mcmaster.ca}}}
\affil[1]{Department of Mathematics, The University of Hong Kong}
\affil[2]{Electrical and Computer Engineering, McMaster University}

\date{}

\begin{document}

\maketitle

\begin{abstract}
The information bottleneck (IB) seeks a representation $U$ of a source $X$ that retains as much information as possible about a target $Y$, subject to a constraint on $I(U;X)$. A classical argument shows that it suffices to consider representations with at most $|\mathcal{X}|+1$ symbols, and this bound is known to be tight whenever $|\mathcal{X}| \geq 3$. We show that the binary case behaves differently: if $X$ is binary and $Y$ is finite, then for every joint distribution of $(X,Y)$ and every rate constraint, the IB optimum is attained by a binary $U$. Hence the bound $|\mathcal{U}| \leq |\mathcal{X}|+1$ sharpens to $|\mathcal{U}| \leq |\mathcal{X}|$ for binary sources. The proof combines a separating hyperplane argument with the observation that, for a binary source, the ratio of the second derivatives of the two entropy functions involved is concave.

\end{abstract}

\section{Introduction}
\label{sec: introduction}

The information bottleneck (IB) \cite{tishby1999information, zaidi2020information} formalizes the extraction of relevant information: given a pair $(X,Y)$ on finite alphabets, one looks for a representation $U$ of $X$, with $U \leftrightarrow X \leftrightarrow Y$, that is as informative about $Y$ as possible while keeping $I(U;X)$ small. Beyond its information-theoretic interest, the IB has served as a lens for analyzing deep neural networks \cite{tishby2015deep, shwartzziv2017opening} and as an objective for learning compact yet informative representations. For $R \geq 0$, the optimal tradeoff is described by the IB function
\begin{equation}
\label{eq: IB function}
    F(R) = \max \big\{ I(U;Y): U \leftrightarrow X \leftrightarrow Y, \ I(U;X) \leq R \big\},
\end{equation}
which is closely related to the conditional entropy bound of Witsenhausen and Wyner \cite{witsenhausen1975conditional}.

A basic question is how large the alphabet $\mathcal{U}$ must be. Witsenhausen and Wyner \cite{witsenhausen1975conditional} first showed that
\[
|\mathcal{U}| \leq |\mathcal{X}| + 1
\]
suffices; the same bound follows from the support lemma \cite{elgamal2011network}. Gilad-Bachrach et al.\ \cite{giladbachrach2003information} claimed that $|\mathcal{U}| \leq |\mathcal{X}|$ suffices.

Benger, Asoodeh and Chen \cite{benger2023cardinality} pointed out that this claim conflates the IB Lagrangian, for which $|\mathcal{X}|$ symbols do suffice, with the IB function $F$ itself, whose graph may contain linear segments. They further showed that the bound $|\mathcal{X}|+1$ is tight: for Hamming channels with uniform input and $|\mathcal{X}| \geq 3$, $|\mathcal{X}|+1$ symbols are necessary for all rates in some interval $(0, R_c)$.

For a binary target $Y$, Tang and Han \cite{tang2026chernoff} proposed a bottleneck model that replaces $I(U;Y)$ by the Chernoff information between $P_{U|Y=0}$ and $P_{U|Y=1}$, and sharpened its cardinality bound from $|\mathcal{X}|+1$ to $k+1$, where $k \leq |\mathcal{X}|$ is the number of distinct likelihood ratios $P_{X|Y=1}(x)/P_{X|Y=0}(x)$. The same argument yields this bound for the classical IB with binary $Y$.

The binary case $|\mathcal{X}| = 2$ is not covered by these results. For a doubly symmetric binary source, Mrs.\ Gerber's Lemma \cite{wyner1973theorem} shows that a binary $U$ is optimal, but this relies on the symmetry of the channel. In this paper we show that, for a binary source, a binary $U$ is always optimal:

\begin{itemize}[leftmargin=*]
    \item For every joint distribution $P_{XY}$ with $\mathcal{X} = \{0,1\}$ and $\mathcal{Y}$ finite, and every $R \geq 0$, the maximum in \eqref{eq: IB function} is attained by some $U$ with $|\mathcal{U}| \leq 2$ (Theorem \ref{thm: main}).
    \item Consequently, the bound $|\mathcal{X}| + 1$ is never tight for binary sources, and the condition $|\mathcal{X}| \geq 3$ in \cite{benger2023cardinality} cannot be dropped.
\end{itemize}

The proof proceeds as follows. We transform the original problem into minimizing $H(Y|U)$ over the convex hull of a curve. By the separating hyperplane theorem, every atom of an optimal $U$ becomes a minimizer of a function $\psi$. We prove that $\psi$ either has at most two minimizers or is constant. In the first case, we merge the symbols of $U$ into a binary $U$ that remains optimal; in the second case, we find a binary optimal $U$ by time-sharing. The key step is a curvature computation that is specific to binary sources.

The rest of the paper is organized as follows. Section \ref{sec: setup} states the setup and the main result, Section \ref{sec: preliminaries} collects two lemmas, Section \ref{sec: proof} proves the main result, and Section \ref{sec: discussion} explains why the argument does not extend to $|\mathcal{X}| \geq 3$ and concludes the paper.

\section{Problem Setup and Main Result}
\label{sec: setup}

Throughout, entropies are in nats and $H(\cdot)$ also denotes the entropy of a probability mass function; we follow the notation of \cite{cover2006elements}.

\begin{itemize}
    \item Let $\mathcal{X} = \{0,1\}$, and let $P_{X} = (1-p, p)$ with $0 < p < 1$ (otherwise everything is trivial).

    \item Let $W = P_{Y|X}$, and denote its rows by $w_{0} = P_{Y|X=0}$ and $w_{1} = P_{Y|X=1}$.

    \item For $0 \leq q \leq 1$, define
    \[
    f(q) = (1-q) w_{0} + q w_{1},
    \]
    with entries $f_y(q) = (1-q) w_0(y) + q w_1(y)$, $y \in \mathcal{Y}$.

    \item Let $h(q)$ be the binary entropy of $q$, and $g(q) = H(f(q))$.

    \item Let $\gamma(q) = (q, h(q), g(q)): [0,1] \to \R^3$, and let $K$ be its convex hull,
    \[
    K = \mathrm{conv} \ \gamma([0,1]) = \Big\{ \sum_{i} \lambda_{i} \gamma(q_{i}): \lambda_{i} \geq 0, \ \sum_{i} \lambda_{i} = 1, \ q_{i} \in [0,1] \Big\}.
    \]
\end{itemize}

The maximum in \eqref{eq: IB function} is taken over $U$ with finite alphabet $\mathcal{U}$. By the general cardinality bound \cite{witsenhausen1975conditional}, it is attained by some $U$ with $|\mathcal{U}| \leq |\mathcal{X}| + 1$. Our main result sharpens this bound in the binary case.

\begin{theorem}
\label{thm: main}
Let $\mathcal{X} = \{0,1\}$ and let $\mathcal{Y}$ be finite. Then, for any $R \geq 0$, the maximum in \eqref{eq: IB function} is attained by some $U$ with $|\mathcal{U}| \leq 2$. That is, the general cardinality bound $|\mathcal{U}| \leq |\mathcal{X}| + 1$ sharpens to $|\mathcal{U}| \leq |\mathcal{X}|$ in the binary case.
\end{theorem}

\section{Preliminaries}
\label{sec: preliminaries}

\begin{lemma}[{Separating hyperplane theorem \cite[Section~2.5]{boyd2004convex}}]
\label{lem: separating hyperplane theorem}
    Let $S_1, S_2 \subset \R^{n}$ be nonempty, convex, and disjoint. Then there exist a nonzero $c \in \R^{n}$ and $\alpha \in \R$ such that
    \[
    c \cdot x \geq \alpha \geq c \cdot y, \quad \forall x \in S_1, \ y \in S_2.
    \]
\end{lemma}

\begin{lemma}[{Identity theorem for polynomials \cite[Chapter~IV]{lang2002algebra}}]
\label{lem: identity theorem for polynomials}
    Let $P(x)$ be a polynomial. If $P(x) = 0$ on a nondegenerate interval $I \subset \R$, then $P(x) \equiv 0$ on $\R$.
\end{lemma}

\section{Proof of Theorem \ref{thm: main}}
\label{sec: proof}

\subsection{Step 1: Reduction to $K$}

Consider a Markov chain $U \leftrightarrow X \leftrightarrow Y$ with $\mathcal{U}$ finite, and let $q_{u} = P(X=1|U=u)$. Then, by Markovity,
\[
P_{Y|U=u} = P_{Y|X=0}P_{X=0|U=u} + P_{Y|X=1}P_{X=1|U=u} = w_{0} (1-q_{u}) + w_{1} q_{u} = f(q_{u}).
\]
Since $P_{X|U=u} = (1-q_{u}, q_{u})$ and $\sum_{u} P_{U}(u) q_{u} = p$, we derive
\[
\sum_{u} P_{U}(u) \gamma(q_{u}) = \big(p, H(X|U), H(Y|U)\big) \in K.
\]

Conversely, let $\lambda_{u} \geq 0$ and $q_{u} \in [0,1]$, $u = 1, \dots, |\mathcal{U}|$, satisfy
\[
\sum_{u} \lambda_{u} = 1, \qquad \sum_{u} \lambda_{u} q_{u} = p.
\]
Setting
\[
P_{U}(u) = \lambda_{u}, \qquad P_{X|U=u} = (1-q_{u}, q_{u}), \qquad P_{Y|X,U} = W
\]
gives
\[
\big(p, H(X|U), H(Y|U)\big) = \sum_{u} \lambda_{u} \gamma(q_{u}).
\]
Therefore, the achievable pairs $(H(X|U), H(Y|U))$ are exactly the points $(s,t)$ with $(p,s,t) \in K$. Moreover, since
\[
I(U;X) = h(p) - H(X|U), \qquad I(U;Y) = H(Y) - H(Y|U),
\]
and both $h(p)$ and $H(Y)$ are fixed, computing \eqref{eq: IB function} is equivalent to
\[
\min \big\{t: (p,s,t) \in K, \ s \geq h(p) - R \big\}.
\]

Let $U^*$ be an optimal $U$ with $|\mathcal{U}| \leq |\mathcal{X}| + 1$, which exists by the general cardinality bound. Without loss of generality, $P_{U^*}(u) > 0$ for all $u$. Let
\[
z^* \stackrel{\Delta}{=} (p, s^*, t^*) = \big(p, H(X|U^*), H(Y|U^*)\big) \in K
\]
be its point, and let $q_{u}^* = P(X=1|U^*=u)$.

\subsection{Step 2: Separating Hyperplane}

Define
\[
S_{t^*} \stackrel{\Delta}{=} \big\{ (p,s,t): s \geq h(p) - R, \ t < t^* \big\} \subset \R^3.
\]
Then $S_{t^*}$ is convex, since it is a product of intervals.

We first prove $K \cap S_{t^*} = \emptyset$. Suppose some $(p,s,t) \in K \cap S_{t^*}$; then $s \geq h(p) - R$ and $t < t^*$. By Step 1, there is a $U$ such that $H(X|U) = s$ and $H(Y|U) = t$, which implies
\[
I(U;X) = h(p) - s \leq R, \qquad H(Y|U) = t < t^*, \ \text{ i.e., } \ I(U;Y) > I(U^*;Y).
\]
This contradicts the optimality of $U^*$. Hence $K \cap S_{t^*} = \emptyset$.

We then prove $z^* \in \overline{S_{t^*}}$. Define
\[
z^*_{n} \stackrel{\Delta}{=} \big(p, s^*, t^* - \tfrac{1}{n}\big), \quad n \geq 1.
\]
Because $U^*$ is feasible, we have $s^* \geq h(p) - R$; and because $t^* - \tfrac{1}{n} < t^*$, we derive $z^*_{n} \in S_{t^*}$. Since $z^*_{n} \to z^*$, we have $z^* \in \overline{S_{t^*}}$.

Applying the separating hyperplane theorem (Lemma \ref{lem: separating hyperplane theorem}) with $S_1 = K$ and $S_2 = S_{t^*}$, there exist $(c_0, c_1, c_2) \neq 0$ and $\alpha \in \R$ such that
\begin{equation}
\label{eq: separating inequality}
    c_0 q + c_1 s + c_2 t \geq \alpha \geq c_0 p + c_1 s' + c_2 t', \quad \forall (q,s,t) \in K, \ (p,s',t') \in S_{t^*}.
\end{equation}
Since $z^* \in K$, the left inequality gives
\[
c_0 p + c_1 s^* + c_2 t^* \geq \alpha;
\]
since $z^* \in \overline{S_{t^*}}$, the right inequality gives
\[
c_0 p + c_1 s^* + c_2 t^* \leq \alpha.
\]
Hence
\[
\alpha = c_0 p + c_1 s^* + c_2 t^*.
\]

We prove $c_2 \geq 0$. Setting $s' = s^*$ and $t' = t^* - n$ in \eqref{eq: separating inequality} gives
\[
\alpha \geq c_0 p + c_1 s^* + c_2 t^* - c_2 n = \alpha - c_2 n,
\]
i.e., $c_2 n \geq 0$ for all $n \geq 1$. Hence $c_2 \geq 0$.

Furthermore, we prove that $c_1$ and $c_2$ cannot both be $0$. Suppose $c_1 = c_2 = 0$. Then the left inequality of \eqref{eq: separating inequality} gives
\[
c_0 q \geq \alpha = c_0 p, \quad \text{i.e.,} \quad c_0 (q-p) \geq 0, \quad \forall (q,s,t) \in K.
\]
Taking $(q,s,t) = \gamma(0)$ and $\gamma(1)$ gives
\[
-c_0 p \geq 0, \qquad c_0 (1-p) \geq 0.
\]
Because $0 < p < 1$, we have $c_0 = 0$, which contradicts $(c_0, c_1, c_2) \neq 0$.

Let
\[
\psi(q) \stackrel{\Delta}{=} c_0 q + c_1 h(q) + c_2 g(q).
\]
Applying the left inequality of \eqref{eq: separating inequality} to $\gamma(q) \in K$ gives
\[
\psi(q) \geq \alpha, \quad \forall q \in [0,1].
\]
Averaging over $U^*$ gives
\begin{align*}
    \sum_{u} P_{U^*}(u) \big(\psi(q_{u}^*) - \alpha\big)
    & = c_0 \sum_{u} P_{U^*}(u) q_{u}^* + c_1 \sum_{u} P_{U^*}(u) h(q_{u}^*) + c_2 \sum_{u} P_{U^*}(u) g(q_{u}^*) - \alpha \\
    & = c_0 p + c_1 s^* + c_2 t^* - \alpha \\
    & = 0.
\end{align*}
Since every term on the left is nonnegative, $\psi(q_{u}^*) = \alpha$ for every $u$. In a word, every $q_{u}^*$ is a minimizer of $\psi$ on $[0,1]$, with minimum value $\alpha$.

\subsection{Step 3: Curvature of $\psi$}

The aim of this step is to show that the set where $\psi'' < 0$ is an interval, and Step 4 uses this to bound the number of minimizers of $\psi$.

Throughout this step, $q \in (0,1)$. Without loss of generality, $P_{Y}(y) > 0$ for all $y \in \mathcal{Y}$, which gives
\[
f_y(q) = (1-q) w_0(y) + q w_1(y) > 0, \quad \forall y \in \mathcal{Y}, \ q \in (0,1).
\]
Let $d \stackrel{\Delta}{=} w_1 - w_0$, so that $f_y(q) = w_0(y) + q\, d(y)$. A direct calculation gives
\[
h''(q) = -\frac{1}{q(1-q)}, \qquad g''(q) = -\sum_{y} \frac{d(y)^2}{f_y(q)}.
\]
Since $(c_0 q)'' = 0$, it follows that
\[
\psi''(q) = c_1 h''(q) + c_2 g''(q) = h''(q)\big(c_1 + c_2 r(q)\big),
\]
where
\[
r(q) \stackrel{\Delta}{=} \frac{g''(q)}{h''(q)} = q(1-q) \sum_{y} \frac{d(y)^2}{f_y(q)}.
\]

We now show that $r$ is concave. For each $y \in \mathcal{Y}$, let $\tilde{f}_y(q) \stackrel{\Delta}{=} (1-q) w_1(y) + q w_0(y) = w_1(y) - q\, d(y)$. Then
\[
f_y(q) \tilde{f}_y(q) = \big(w_0(y) + q d(y) \big) \big(w_1(y) - q d(y) \big) = w_0(y) w_1(y) + d(y)^2 q (1-q).
\]
Dividing by $f_y(q)$, summing over $y$, and using $\sum_{y} \tilde{f}_y(q) = 1$, we obtain
\begin{equation}
\label{eq: equation of r(q)}
    r(q) = 1 - \sum_{y} \frac{w_0(y) w_1(y)}{f_y(q)}.
\end{equation}
Each $1/f_y(q)$ is convex in $q$, being the reciprocal of a positive affine function, and $w_0(y) w_1(y) \geq 0$. Hence $r$ is concave on $(0,1)$.

Finally, since $h''(q) < 0$,
\[
\mathcal{Q} \stackrel{\Delta}{=} \big\{q \in (0,1): \psi''(q) < 0 \big\} = \big\{q \in (0,1): c_1 + c_2 r(q) > 0 \big\}.
\]
Since $c_2 \geq 0$ (Step 2), the mapping $q \mapsto c_1 + c_2 r(q)$ is concave on $(0,1)$. Hence $\mathcal{Q}$, a superlevel set of a concave function, is an interval, possibly empty.

\subsection{Step 4: At Most Two Minimizers}

With $\psi$ and $\alpha$ defined in Step 2, we have the following lemma.

\begin{lemma}
\label{lem: either at most 2 or equiv}
    If $\psi(q) = \alpha$ holds for more than two $q \in [0,1]$, then $\psi \equiv \alpha$ on $[0,1]$.
\end{lemma}

\begin{proof}
    Recall from Step 2 that $\psi \geq \alpha$ on $[0,1]$. Moreover, $\psi$ is continuous on $[0,1]$ and twice differentiable on $(0,1)$. Suppose there exist $0 \leq q_1 < q_2 < q_3 \leq 1$ such that
    \[
    \psi(q_{1}) = \psi(q_{2}) = \psi(q_{3}) = \alpha.
    \]

    With $\mathcal{Q}$ defined in Step 3, we first prove that $\mathcal{Q}$ cannot meet both intervals $(q_1, q_2)$ and $(q_2, q_3)$. If it does, then, since $\mathcal{Q}$ is an interval, we have $q_2 \in \mathcal{Q}$, hence $\psi''(q_2) < 0$. However, $q_2 \in (0,1)$ and $\psi(q_{2}) = \alpha = \min \psi$, so $q_2$ is an interior minimizer of $\psi$, which forces $\psi''(q_2) \geq 0$, a contradiction.

    Therefore, $\mathcal{Q}$ misses one of $(q_1, q_2)$ and $(q_2, q_3)$. Assume it misses $(q_1, q_2)$; the other case is identical. Then $\psi''(q) \geq 0$ on $(q_1, q_2)$, so $\psi$ is convex on $[q_1, q_2]$. Since $\psi(q_{1}) = \psi(q_{2}) = \alpha$, convexity and continuity give $\psi(q) \leq \alpha$ on $[q_1, q_2]$. Combining with $\psi \geq \alpha$, we have
    \[
    \psi(q) = \alpha, \quad \forall q \in [q_1, q_2].
    \]
    Hence $\psi''(q) = 0$ on $(q_1, q_2)$. Since $\psi''(q) = h''(q)(c_1 + c_2 r(q))$ with $h''(q) \neq 0$, we derive
    \[
    c_1 + c_2 r(q) = 0, \quad \forall q \in (q_1, q_2).
    \]

    Define
    \[
    \mathcal{G}(q) \stackrel{\Delta}{=} \big(c_1 + c_2 r(q) \big) \prod_{y} f_y(q).
    \]
    Substituting \eqref{eq: equation of r(q)} gives
    \begin{align*}
    \mathcal{G}(q) & = \Big(c_1 + c_2 - c_2 \sum_y \frac{w_0(y) w_1(y)} {f_y(q)}\Big) \prod_{y'} f_{y'}(q) \\
    & = (c_1 + c_2) \prod_{y'}  f_{y'}(q) - c_2 \sum_{y} w_0(y) w_1(y) \frac{\prod_{y'}f_{y'}(q)} {f_y(q)} \\
    & = (c_1 + c_2) \prod_{y'} f_{y'}(q) - c_2 \sum_{y} w_0(y) w_1(y) \prod_{y'\neq y} f_{y'}(q).
    \end{align*}
    Since each $f_y(q)$ is affine in $q$, the function $\mathcal{G}(q)$ is a polynomial.

    Since $\mathcal{G}(q) = 0$ on the nondegenerate interval $(q_1, q_2)$, the identity theorem for polynomials (Lemma \ref{lem: identity theorem for polynomials}) gives $\mathcal{G}(q) \equiv 0$. Because $\prod_{y} f_y(q) > 0$ on $(0,1)$, we have $c_1 + c_2 r(q) \equiv 0$ on $(0,1)$, which gives
    \[
    \psi''(q) = h''(q)\big(c_1 + c_2 r(q)\big) = 0, \quad \forall q \in (0,1).
    \]
    Hence, for some $\zeta_0,\zeta_1 \in \R$,
    \[
    \psi(q)=\zeta_1q+\zeta_0 \quad \text{on } (0,1),
    \]
    and by continuity also on $[0,1]$. Since $\psi(q_1) = \psi(q_2) = \alpha$ and $q_1 < q_2$, we derive $\zeta_1 = 0$ and $\zeta_0 = \alpha$, i.e., $\psi \equiv \alpha$ on $[0,1]$.
\end{proof}

\subsection{Step 5: Construction of a Binary Optimal $U$}

By Step 2, every $q_u^*$ is a minimizer of $\psi$, i.e., $\psi(q_u^*) = \alpha$. Hence, by Lemma \ref{lem: either at most 2 or equiv}, one of two cases holds: either the $q_u^*$ take at most two distinct values, or $\psi \equiv \alpha$.

\paragraph{Case 1: the $q_u^*$ take at most two distinct values.} Let the distinct values be $\hat{q}_1, \dots, \hat{q}_{k}$, with $k \in \{1, 2\}$. Merging the symbols of $U^*$ that share the same $q_u^*$ gives a new $\hat{U} = \sigma(U^*)$, where
\[
\sigma(u) = j, \quad \text{if} \ q^*_{u} = \hat{q}_j, \ j=1, \dots, k.
\]
We now prove that $\hat{U}$ has the same point $z^*$. Since $\hat U$ is a function of $U^*$, we still have the Markov chain $\hat{U} \leftrightarrow X \leftrightarrow Y$. For $j = 1, \dots, k$,
\[
P_{\hat U}(j) = \sum_{u: \, q_u^* = \hat{q}_j} P_{U^*}(u),
\qquad
P(X=1 | \hat{U} = j) = \frac{1}{P_{\hat U}(j)} \sum_{u: \, q_u^* = \hat{q}_j} P_{U^*} (u) q_u^* = \hat{q}_j.
\]
Then, by Step 1,
\begin{align*}
\big(p, H(X|\hat U), H(Y|\hat U)\big)
& = \sum_j P_{\hat U}(j) \gamma(\hat{q}_j) \\
& = \sum_j \sum_{u: \, q_u^* = \hat{q}_j} P_{U^*}(u) \gamma(q_u^*) \\
& = \sum_u P_{U^*}(u) \gamma(q_u^*) \\
& = z^*.
\end{align*}
Therefore, $\hat U$ is optimal, with $|\hat{\mathcal U}|\leq 2$.

\paragraph{Case 2: $\psi \equiv \alpha$.} Then $\psi''(q) \equiv 0$ on $(0,1)$. Since $\psi''(q) = h''(q) (c_1 + c_2 r(q))$ with $h''(q) \neq 0$, we get $c_1+c_2r(q) \equiv 0$. In Step 2 we proved that $c_2 \geq 0$ and that $c_1, c_2$ cannot both be $0$. If $c_2 = 0$, then $c_1 + c_2 r(q) \equiv 0$ forces $c_1 = 0$, a contradiction. Thus $c_2 > 0$.

For any $U \leftrightarrow X \leftrightarrow Y$, Step 1 gives
\begin{align*}
    c_0 p + c_1 H(X|U) + c_2 H(Y|U) = \sum_u P_U(u) \psi(q_u) = \alpha.
\end{align*}
Since $c_2>0$, this gives
\[
H(Y|U) = \frac{\alpha - c_0 p - c_1 H(X|U)}{c_2},
\]
so $H(Y|U)$ is determined by $H(X|U)$. Hence any $U$ with $H(X|U) = s^*$ has $H(Y|U) = t^*$, i.e.,
\[
I(U;X) = I(U^*;X) \leq R, \qquad I(U;Y) = I(U^*;Y),
\]
so it is optimal.

It remains to find a binary $U$ with $H(X|U)=s^*$, which we do by time-sharing. For $\theta \in [0,1]$, let $T \sim \mathrm{Bern} (\theta)$ and $X' \sim P_X$ be independent of each other and of $(X,Y)$, and define
\[
U_{\theta} = \begin{cases} X, & T = 1,\\ X', & T = 0. \end{cases}
\]
Since $X$ and $X'$ are binary, $U_{\theta}$ is binary. Since $U_{\theta}$ is a function of $X$ and the pair $(T,X')$, which is independent of $(X,Y)$, we have $U_{\theta} \leftrightarrow X\leftrightarrow Y$. Moreover, conditioning on $T$ gives
\begin{align*}
    P(X=x, U_{\theta} = u) & = \theta P(X=x, X=u) + (1 - \theta) P(X=x, X'=u) \\
    & = \theta P_X(x) \mathbf{1} \{u=x\} + (1 - \theta) P_X(x) P_X(u),
\end{align*}
which is affine in $\theta$. Since $H(X|U_{\theta})=H(X,U_{\theta})-H(U_{\theta})$ is a continuous function of this pmf, $H(X|U_{\theta})$ is continuous in $\theta$. At the endpoints,
\[
H(X|U_1) = H(X|X) = 0, \qquad H(X|U_0) = H(X|X') = H(X) = h(p),
\]
since $U_1 = X$ and $U_0 = X'$ is independent of $X$. Since
\[
0 \leq s^* = H(X|U^*) \leq H(X) = h(p),
\]
the intermediate value theorem gives $\theta^*\in[0,1]$ with $H(X|U_{\theta^*})=s^*$. By the above, $U_{\theta^*}$ is optimal and binary.

In both cases, the maximum in \eqref{eq: IB function} is attained by some $U$ with $|\mathcal{U}| \leq 2$, which completes the proof of Theorem \ref{thm: main}. \hfill $\blacksquare$

\section{Discussion and Conclusion}
\label{sec: discussion}

We showed that, in the information bottleneck with a binary source $X$ and a finite target $Y$, a binary representation $U$ is always optimal, for every joint distribution and every rate constraint. Thus the general cardinality bound $|\mathcal{U}| \leq |\mathcal{X}|+1$ sharpens to $|\mathcal{U}| \leq |\mathcal{X}|$ in the binary case. We now explain why the argument does not extend to $|\mathcal{X}| \geq 3$.

For $|\mathcal{X}| \geq 3$, Steps 1 and 2 still go through. Each symbol $u$ corresponds to a point $\mathbf{q}_u = P_{X|U=u}$ of the probability simplex on $\mathcal{X}$, which has dimension $|\mathcal{X}|-1 \geq 2$, and every $\mathbf{q}_u$ of an optimal $U$ minimizes, on this simplex, a function of the form
\[
\psi(\mathbf{q}) = \ell(\mathbf{q}) + c_1 H(\mathbf{q}) + c_2 H(\mathbf{q} W),
\]
with $\ell$ affine and $c_2 \geq 0$. Steps 3 and 4, however, rely on $q$ being a scalar:

\begin{itemize}[leftmargin=*]
    \item \emph{Curvature.} In the binary case, $\psi'' = h''(c_1 + c_2 r)$ with $r$ concave, so the set where $\psi'' < 0$ is an interval. For $|\mathcal{X}| \geq 3$, the second derivative becomes the Hessian
    \[
    c_1 \nabla^2_{\mathbf{q}} H(\mathbf{q}) + c_2 \nabla^2_{\mathbf{q}} H(\mathbf{q}W),
    \]
    and the set where it fails to be positive semidefinite is determined by a generalized eigenvalue problem rather than by the sign of a single concave function.
    \item \emph{Ordering.} Step 4 uses that, among three minimizers $q_1 < q_2 < q_3$, the middle one lies between the other two. A simplex of dimension at least two has no such ordering, and the target would be at most $|\mathcal{X}|$ minimizers rather than two.
\end{itemize}

These obstructions are not merely technical. For the Hamming channels in \cite{benger2023cardinality}, every optimal $U$ needs $|\mathcal{X}|+1$ symbols at small rates. Together with Theorem \ref{thm: main}, this shows that the tight cardinality bound, i.e., the smallest $M$ such that $|\mathcal{U}| \leq M$ suffices for every $P_{XY}$, is
\[
M = \begin{cases}
    2, & |\mathcal{X}| = 2, \\
    |\mathcal{X}| + 1, & |\mathcal{X}| \geq 3.
\end{cases}
\]
Since symbols with the same $\mathbf{q}_u$ can be merged as in Step 5, the corresponding $\psi$ for these Hamming channels has at least $|\mathcal{X}|+1$ distinct minimizers; hence no argument of this type can yield $|\mathcal{U}| \leq |\mathcal{X}|$ for all sources with $|\mathcal{X}| \geq 3$. Characterizing the sources with $|\mathcal{X}| \geq 3$ for which $|\mathcal{X}|$ symbols suffice remains an open question.

\end{document}